\documentclass[12pt]{article}

\usepackage{amsmath}
\usepackage{amssymb}
\usepackage{amsfonts}
\usepackage{latexsym}
\usepackage{color}

\catcode `\@=11 \@addtoreset{equation}{section}

\catcode `\@=12

\newcommand{\be}{\begin{equation}}
	\newcommand{\en}{\end{equation}}
\newcommand{\bea}{\begin{eqnarray}}
	\newcommand{\ena}{\end{eqnarray}}
\newcommand{\beano}{\begin{eqnarray*}}
	\newcommand{\enano}{\end{eqnarray*}}
\newcommand{\bee}{\begin{enumerate}}
	\newcommand{\ene}{\end{enumerate}}

\newcommand{\mb}{\mathbb}

\newcommand{\B}{{\mathfrak B}}

\newcommand{\mc}{\mathcal}
\newcommand{\norm}[1]{ \parallel #1 \parallel}

\newcommand{\id}{{\mb I}}

\newcommand{\Sc}{{\cal S}}

\newcommand{\F}{{\cal F}}

\newcommand{\C}{{\cal C}}
\newcommand{\1}{1 \!\! 1}

\newcommand{\Hil}{\mc H}
\newcommand{\xhp}{x_{\hat\Psi}(t)}
\newcommand{\dhp}{\delta_{\hat\Psi}}

\renewcommand{\l}{\langle}
\renewcommand{\r}{\rangle}
\newcommand{\pint}{\l\cdot,\cdot\r}
\newcommand{\pin}[2]{\l#1 , #2\r}
\newcommand{\nor}{\|\cdot\|}

\newtheorem{thm}{Theorem}

\newtheorem{lemma}[thm]{Lemma}
\newtheorem{prop}[thm]{Proposition}
\newtheorem{defn}[thm]{Definition}

\newenvironment{proof}{\noindent {\bf Proof --}}{\hfill$\square$ \vspace{3mm}\endtrivlist}

\catcode `\@=11 \@addtoreset{equation}{section}
\catcode `\@=12

\begin{document}

\thispagestyle{empty}

\vspace*{2cm}

\begin{center}
{\Large \bf New results for Heisenberg dynamics for non self-adjoint Hamiltonians}   \vspace{2cm}\\

{\large F. Bagarello}\\
Dipartimento di Ingegneria,
Universit\`a di Palermo,\\ I-90128  Palermo, Italy\\
and I.N.F.N., Sezione di Catania\\
e-mail: fabio.bagarello@unipa.it\\

\end{center}

\vspace*{1cm}

\begin{abstract}

In a previous paper we  began our analysis on the role of non self-adjoint Hamiltonians in connection with the  Heisenberg dynamics for quantum systems. Here, motivated by the growing interest on this topic and on some recent results on dynamical systems, we continue this analysis focusing on what we believe is an unexplored (or, at least, not so explored! aspect of Heisenberg dynamics, related to the need for using vectors which are {\em brute-force normalized}. Our main interest is on conserved quantities, and on conditions which guarantee that some observables of the system, or their mean values, do not evolve in time.

\end{abstract}

\vspace{2cm}


\vfill


\newpage

\section{Introduction}

The relevance of non-Hermitian\footnote{In this paper we will use Hermitian and self-adjoint as sinonimous terms.} Hamiltonians (or, more generally, non-Hermitian observables) in quantum mechanics is nowadays out of discussion. Their role for some physical systems is essential, and the mathematical consequences of their presence produce quite interesting and highly non trivial results.  Some relevant monographs and edited books on several aspects of these observables are  \cite{benbook}-\cite{bagspringer}.

In this paper we will focus on a very special aspect of non-Hermitian Hamiltonians: the Heisenberg dynamics generated by some $H\neq H^\dagger$, on the same lines as our original analysis, began in \cite{bagAoP1,bagAoP2,bag2022}, where we considered some aspects of this dynamics, showing in particular, \cite{bag2022}, the kind of technical problems that arise when $H\neq H^\dagger$, since the time evolution of the product of two observables $X$ and $Y$, $(XY)(t)$, is in general different from $X(t)Y(t)$. In \cite{bag2022} we have also proposed a possible definition of what a {\em symmetry} and an {\em integral of motion} should be for a system $\Sc$ driven by some $H\neq H^\dagger$. This particular aspect is rather important and was already considered by different authors, see \cite{yog1}-\cite{lisok} for instance, but in \cite{bag2022} we proposed some systematic approach to the problem. 

In a different context, that of dynamical systems and Decision Making, in an older paper, \cite{bag2020}, in the attempt to describe a {\em one-directional dynamics}, we have proposed some very simple dynamical systems driven by non-Hermitian Hamiltonians where, in some unexpected way, conserved quantities are found. These conserved quantities are the mean values of certain global number operators relevant for the analysis of the system $\Sc$ computed not on the time evoluted wave function $\Psi(t)$ of $\Sc$, but on its normalized version, $\hat\Psi(t)=\frac{\Psi(t)}{\norm{\Psi(t)}}$. The reason for this choice, \cite{muga,sim,sgh}, is that when $H\neq H^\dagger$, the norm of the wave function of $\Sc$ is not preserved in time, if we use $\Psi(t)=e^{-iHt}\Psi(0)$, so that it can easily happen that $\|\Psi(t)\|$ explodes or goes to zero with $t$. In some situations this is exactly what we expect, but in other cases this is something we want to avoid. And a way to avoid this behavior is to replace $\Psi(t)$ with $\hat\Psi(t)$. This replacement changes quite a bit also the mathematics of the problem, and this is indeed what we will discuss in this paper, both proposing some general results and describing a concrete example, borrowed from \cite{bag2020}.

The paper is organized as follows:

In the next section we introduce the notation and list some preliminary results, most of them already deduced in \cite{bag2022}.

Then, in Section \ref{sect3}, we discuss what happens when $\Psi(t)$ is replaced by $\hat\Psi(t)$, both at the level of Schr\"odinger and of Heisenberg dynamics. In particular,  we focus on the existence of conserved quantities. We will see that several differences arise with respect to what happens when using $\Psi(t)$. Hence, normalization is not a trivial procedure for us.

In Section \ref{sect4} we discuss an example coming from Decision Making, \cite{bag2020}, and we show that what was found in \cite{bag2020} fits well in the results of Section \ref{sect3}.

Section  \ref{sectconcl} contains our conclusions, and plans for the future.

  \section{Preliminaries}\label{sect2}
  
  In this section we will introduce first the mathematical settings we work with along this paper, and then we will briefly resume some of our previous results in \cite{bag2022}.
  
  \subsection{The mathematical settings}\label{sect2a}

  Let $\Hil$ be an Hilbert space, with $dim(\Hil)=N<\infty$, with scalar product $\pint$ and associated norm $\nor$, where $\|f\|=\sqrt{\pin{f}{f}}$, $\forall f\in\Hil$. As usual we have $\left<f,g\right>=\sum_{k=1}^N\overline{f_k}\,g_k$, $f,g\in\Hil$. The scalar product is also linked to the  conjugation $\dagger$ defined as $\pin{X^\dagger f}{g}=\pin{f}{Xg}$, $\forall f,g\in\Hil$. Here $X$ is an operator on $\Hil$ which, in our particular case, is a square matrix of dimension $N\times N$. More in general, also in view of future extensions to infinite-dimensional Hilbert spaces, we could say that $X\in\B(\Hil)$, the $C^*$-algebra of all bounded operators on $\Hil$, \cite{br}. In doing this the dimensionality of $\Hil$ could also be infinite.

  The main ingredient of our analysis is an operator (i.e. a matrix) $H$, acting on $\Hil=\mathbb{C}^{N}$, with $H\neq H^\dagger$. We assume for simplicity (but this does not affect our conclusions) that $H$ has exactly $N$ distinct eigenvalues $E_n$, $n=1,2,\ldots,N$. Here, the adjoint $H^\dagger$ of $H$ is the usual one, i.e. the complex conjugate of the transpose of the matrix $H$. 
  
  Other than being all different, we will also assume that $E_n\in\mathbb{R}$, $n=1,2,,\ldots,N$. This is more important since, in this way, we can check that $H$ and $H^\dagger$ are isospectral and, as such, they admit an intertwining operator $X$ such that $XH=H^\dagger X$. This isospectrality is lost if even only one eigenvalue of $H$ is complex. This situation has been considered, for instance, in \cite{bagAoP1,bagAoP2}, to which we refer for more details.

 The $N$ distinct real eigenvalues of $H$ correspond to $N$ distinct eigenvectors $\varphi_k$, $k=1,2,\ldots,N$:
  \be
  H\varphi_k=E_k\varphi_k.
  \label{21}\en
  The set $\F_\varphi=\{\varphi_k,\,k=1,2,\ldots,N\}$ is a basis for $\Hil$, since the eigenvalues are all different. Then an unique biorthogonal basis of $\Hil$, $\F_\Psi=\{\Psi_k,\,k=1,2,\ldots,N\}$, surely exists, \cite{chri,heil}: $\left<\varphi_k,\Psi_l\right>=\delta_{k,l}$, for all $k, l$. Moreover, for all $f\in\Hil$, we can write
  $f=\sum_{k=1}^N\left<\varphi_k,f\right>\Psi_k = \sum_{k=1}^N\left<\Psi_k,f\right>\varphi_k$.  We know that the vectors in $\F_\Psi$ are eigenstates of $H^\dagger$ with eigenvalues $E_k$, see e.g. \cite{bagAoP1}:
  \be
  H^\dagger\Psi_k=E_k\Psi_k,
  \label{22}\en
  $k=0,1,2,\ldots,N$. 
  
  Using the bra-ket notation we can write $\sum_{k=0}^N|\varphi_k\left>\right<\Psi_k|=\sum_{k=0}^N|\Psi_k\left>\right<\varphi_k|=\1$, where, for all $f,g,h\in\Hil$, we define $(|f\left>\right<g|)h:=\left<g,h\right>f$, and $\1$ is the identity operator on $\Hil$. We introduce the operators $S_\varphi=\sum_{k=0}^N|\varphi_k\left>\right<\varphi_k|$ and $S_\Psi=\sum_{k=0}^N|\Psi_k\left>\right<\Psi_k|$, as in \cite{bagbook}. These are bounded positive, self-adjoint, invertible operators, one the inverse of the other: $S_\Psi=S_\varphi^{-1}$. They are often called {\em metric operators}, since they can be used to define new scalar products in $\Hil$. This aspect will not be discussed here, since it is not relevant for us. We refer to \cite{bag2022} for some results in this direction. Moreover:
  \be
  S_\varphi\Psi_n=\varphi_n,\quad S_\Psi\varphi_n=\Psi_n,\quad \mbox{ as well as }\quad S_\Psi H=H^\dagger S_\Psi,\quad S_\varphi H^\dagger= HS_\varphi.
  \label{23}\en
 Many other details are discussed in \cite{bagspringer,bagAoP1,bag2022,bagbook}.

 \subsection{Some previous results}\label{sect2b}

 Here what is relevant for us, is to introduce the Heisenberg-like dynamics we will consider next, and discuss its properties. To fix the ideas, we consider the   Schr\"odinger equation $i\dot \Psi(t)=H\Psi(t)$ as the starting point of our analysis. Hence we introduce a map $\gamma^t$ on $\B(\Hil)$ as follows:
 \be
 \pin{\varphi(t)}{X\psi(t)}=\pin{\varphi_0}{\gamma^t(X)\psi_0},
 \label{24}\en
 where $\varphi(t)$ and $\psi(t)$ both obey the Schr\"odinger equation, with initial values  $\varphi(0)=\varphi_0$ and $\psi(0)=\psi_0$. Then we have
 \be
 \gamma^t(X)=e^{iH^\dagger t}Xe^{-iHt}.
 \label{25}\en
 This is what we have called {\em $\gamma$-dynamics} in \cite{bag2022}. It is clear that, $\forall X\in\B(\Hil)$, $\gamma^t(X)\in\B(\Hil)$ as well, $\forall t\in\mathbb{R}$. It is also clear that, if $H=H^\dagger$, then $\gamma^t(X)=e^{iHt}Xe^{-iHt}$, which is exactly the standard Heisenberg dynamics for a self-adjoint Hamiltonian.

 For self-adjoint Hamiltonians many results are well known. Indeed, 
 let us consider an operator $H_0=H_0^\dagger$, independent on time and such that $H_0\in\B(\Hil)$, and let us call
 \be
 \alpha_0^t(X)=e^{iH_0t}Xe^{-iH_0t},
 \label{26}\en
 $X\in\B(\Hil)$. Then we have:
 \begin{enumerate}
 	\item $\alpha_0^t(\1)=\1$, $\forall t\in\mathbb{R}$, where $\1$ is the identity operator on $\Hil$. Hence $\alpha_0^t$ {\em preserves } the identity operator.
 	\item $\alpha_0^t(XY)=\alpha_0^t(X)\alpha_0^t(Y)$, $\forall X,Y\in\B(\Hil)$: $\alpha_0^t$ is an authomorphism of $\B(\Hil)$.
 	\item $\alpha_0^t$ preserves the adjoint: $\alpha_0^t(X^\dagger)= (\alpha_0^t(X))^\dagger$, $\forall X\in\B(\Hil)$.
 	\item $\alpha_0^t$ is norm continuous: if $\{X_n\}$ converges in the norm of $\B(\Hil)$ to $X$, $\|X_n-X\|\rightarrow 0$ for $n\rightarrow\infty$, then $\|\alpha_0^t(X_n)-\alpha_0^t(X)\|\rightarrow 0$ in the same limit.
 	\item if $Z\in\B(\Hil)$ commutes with $H_0$, $[H_0,Z]=0$, then $\alpha_0^t(Z)=Z$,  $\forall t\in\mathbb{R}$: all the operators commuting with $H_0$ are {\em constant of motion}.
 	\item if we introduce $\delta_0(X)=\lim_{t,0}\frac{\alpha_0^t(X)-X}{t}$, the limit to be understood in the norm of $\B(\Hil)$, then $\delta_0$ is a *-derivation: (i) $\delta_0(X^\dagger)=(\delta_0(X))^\dagger$, and (ii) $\delta_0(XY)=X\delta_0(Y)+\delta_0(X)Y$, $\forall X,Y\in\B(\Hil)$.
 	\item The series $\sum_{k=0}^{\infty}\frac{t^k\delta_0^k(X)}{k!}$ is norm convergent to $\alpha_0^t(X)$ for all $X\in\B(\Hil)$. Here $\delta_0^k(X)$ is defined recursively as follows: $\delta_0^0(X)=X$, and $\delta_0^k(X)=\delta_0(\delta_0^{k-1}(X))$, $k\geq1$.
 \end{enumerate}

 Some of these properties of $\alpha_0^t$ can be extended to $\gamma^t$, but some other can not. The proofs and the details of our claims below can be found in \cite{bag2022}. We start recalling that $\gamma^t(XY)=\gamma^t(X)\gamma^t(Y)$, $\forall X,Y\in\B(\Hil)$ if and only if $\gamma^t(\1)=\1$. Then $\gamma^t$ can be an automorphism only if the identity operator does not evolve in time (with respect to $\gamma^t$).
Moreover,  $\gamma^t(\1)=\1$ only if $H=H^\dagger$. Hence  $\gamma^t$ is an automorphism of $\B(\Hil)$ only for self-adjoint $H$. There is no way out.

 Next, the norm continuity of $\gamma^t$ is a consequence of the fact that $H\in\B(\Hil)$. Indeed in this case, since $\|H\|<\infty$,  the series $\sum_{k=0}^\infty \frac{1}{k!}\,(iHt)^k$ is norm-convergent for all $t\in\mathbb{R}$ and this allows us to conclude that
$$
\|e^{\pm iHt}\|\leq e^{|t|\|H\|}, \qquad \|e^{\pm iH^\dagger t}\|\leq e^{|t|\|H\|},
$$
 $\forall t\in\mathbb{R}$. Hence, if $\|X_n-X\|\rightarrow 0$ for $n\rightarrow\infty$, it easily follows that 
 $
 \left\|\gamma^t(X_n)-\gamma^t(X)\right\|\rightarrow0,
 $
 for $n\rightarrow\infty$.

 Pushing forward our parallel between $\alpha_0^t$ and $\gamma^t$, we now define a map $\delta_\gamma:\B(\Hil)\rightarrow\B(\Hil)$ as follows:
 \be
 \delta_\gamma(X)=\|.\|-\lim_{t,0}\frac{\gamma^t(X)-X}{t}=i\left(H^\dagger X-XH\right),
 \label{32a}\en
 $X\in\B(\Hil)$, where the last result follows from the continuity of $H$, $H^\dagger$, and of their exponentials. Notice that $ \delta_\gamma(X)\in\B(\Hil)$. This is what we call {\em $\gamma$-derivation}, which is inner, \cite{bag2022}:  it exists  $H\in\B(\Hil)$ which, together with $H^\dagger$, {\em represents} $\delta_\gamma$ as a  generalized commutator.
 
 We know  that $\delta_\gamma(X^\dagger)=(\delta_\gamma(X))^\dagger$, $\forall X\in\B(\Hil)$, and $\delta_\gamma$ is norm-continuous, since $\|\delta_\gamma(X_n)-\delta_\gamma(X)\|\leq2\|H\|\|X_n-X\|\rightarrow0$, for any sequence $\{X_n\}$ norm-convergent to $X$. Also, putting $\delta_\gamma^0(X)=X$ and $\delta_\gamma^k(X)=\delta_\gamma(\delta_\gamma^{k-1}(X))$, $k\geq1$, we know that
 	the series $\sum_{k=0}^{\infty}\frac{t^k\delta_\gamma^k(X)}{k!}$ is norm convergent to $\gamma^t(X)$, for all $X\in\B(\Hil)$.
 This, in particular, is the counterpart of the claim in 7 in the list above of the results for $H_0$.
Another interesting result, which extends our previous claim, is the following:

 \begin{prop}\label{prop3}
 	The following statements are equivalent: 1) $\delta_\gamma$ is a *-derivation; 2) $\delta_\gamma(\1)=0$; 3) $H=H^\dagger$; 4) $\gamma^t(\1)=\1$; 5)  $\gamma^t(XY)=\gamma^t(X)\gamma^t(Y)$, $\forall X,Y\in\B(\Hil)$.
 \end{prop}

 The main content of this proposition is that $\gamma^t$ cannot be an authomorphism of $\B(\Hil)$ if any of the above properties (and therefore all) is violated. Again, the fact that the identity operator is stable under time evolution looks crucial.
 We refer to \cite{bag2022} for a very {\em minimal} and illustrative example of the kind of problems which naturally arise when $H\neq H^\dagger$: simple dynamical systems cannot be solved easily since the time evolution of products of variables is not, in general, the product of the time evolution of the same variables. This creates a lot of unpleasant consequences, when one tries to find a closed system of differential equations to solve.

Within the contest considered here, in \cite{bag2022} we have also proved the following result, related to the possibility of having conserved quantities:

 \begin{prop}\label{prop4}
 Let  $X\in\B(\Hil)$. The following statements are all equivalent: 1) $H^\dagger X=XH$; 2) $\delta_\gamma(X)=0$; 3) $ \gamma^t(X)=X$. When $X$ satisfies these conditions, $X$ is called a $\gamma$-symmetry.
 \end{prop}
It is  easy to construct, out of a $\gamma$-symmetry $X$, a set of other $\gamma$-symmetries: if $X$ satisfies $H^\dagger X=XH$, then the new operator $X_H:=XH$, which is different from $H$, in general, also satisfies an intertwining relation of the same kind: $H^\dagger X_H=X_HH$. Indeed we have
$$
H^\dagger X_H=H^\dagger X H=XHH=X_HH,
$$ so that $\delta_\gamma(X_H)=0$ and $ \gamma^t(X_H)=X_H$. Hence $X_H$ is also an integral of motion. Of course, we can repeat the procedure by defining next $X_{H^2}=X_HH=XH^2$, and deduce that $X_{H^2}$ is also a $\gamma$-symmetry. And so on. Not all these operators need to be different: for some $k$ and $l$ it might happen that $X_{H^k}=X_{H^l}$, even if $k\neq l$. Moreover, if $\dim(\Hil)=N<\infty$, as it is the case all throughout this paper, not all the $X_{H^k}$ are independent. This is because $H^N$ can be written as a linear combination of $H^0$, $H^1$, $\ldots$, $H^{N-1}$, because of the Hamilton-Caley Theorem, \cite{gant}.

\vspace{3mm}

In view of what we have seen here, and in particular the role of $\gamma^t(\1)$, it is interesting to investigate a bit more the role of the time evolution of the identity operator when $H\neq H^\dagger$. Let us put
\be
I_{\1}(t)=\|\Psi(t)\|^2=\langle\Psi(t),\1\Psi(t)\rangle=\langle\Psi(0),\gamma^t(\1)\Psi(0)\rangle,
\label{28}\en
where we have used $\Psi(t)=e^{-iHt}\Psi(0)$. It is easy to check that
\be
\frac{d}{dt}I_{\1}(t)=i\langle\Psi(t),(H^\dagger-H)\Psi(t)\rangle.
\label{29}\en
This implies that, if $H=H^\dagger$, then $\Psi(t)$ is normalized in time, as expected: $\|\Psi(t)\|=\|\Psi(0)\|=1$, if $\Psi(0)$ was chosen to be normalized. However, the opposite implication is false. We can only prove that, if $\frac{d}{dt}I_{\1}(t)=0$ for all vectors $\Psi(t)$, then all the eigenvalues of $H$ must be real. Indeed, let us suppose this is not so, and that for some $l$ we have $H\varphi_l=E_l\varphi_l$, with $\Im(E_l)\neq0$. Hence, if we take $\Psi(0)=\varphi_l$, we would have 
$$
I_{\1}(t)=\langle e^{-iHt}\varphi_l, e^{-iHt}\varphi_l\rangle=e^{i(\overline{E_l}-E_l)t}\|\varphi_l\|^2,
$$
which is surely dependent on time because of our assumption on $\Im(E_l)$, contrarily to what expected.

Of course, this does not necessarily imply that $H=H^\dagger$. For instance, if $H$ is similar to a self-adjoint Hamiltonian $H_0$, i.e. if it exists an invertible operator $R$, not unitary (to avoid trivial situations) such that $H=RH_0R^{-1}$, then simple computations show that, if $[H_0,R^\dagger R]=0$, we still have
$$
e^{iH^\dagger t}e^{-iHt}=\1, \qquad \Rightarrow \qquad I_{\1}(t)=I_{\1}(0),
$$
for any possible choice of $\Psi(0)$.

\section{The non-linear Hamiltonian}\label{sect3}

In this section we focus on a slightly different situation, closer to what  considered in \cite{bagAoP1,bagAoP2,yog1,sim} for instance, in which that mean values which are interesting for us are not, as in (\ref{28}), of the form $\langle\Psi(t),X\Psi(t)\rangle$ but, rather, of the form $\langle\hat\Psi(t),X\hat\Psi(t)\rangle$, where $\hat\Psi(t)$ is the normalized version of $\Psi(t)$, $\hat\Psi(t)=\frac{\Psi(t)}{\norm{\Psi(t)}}$. As already mentioned in the Introduction, the reason for this interest is based not only on what discussed in the cited papers, but also on our explicit results in \cite{bag2020}. We refer to Section \ref{sect4} for a full analysis of this particular case: we will see that the model proposed in \cite{bag2020} is indeed an explicit non-trivial realization of what we will discuss in this section.

Our first scope here is to find the differential equation for $\hat\Psi(t)$, assuming that $\Psi(t)$ obeys the Schr\"odinger equation $i\dot\Psi(t)=H\Psi(t)$, where $H\neq H^\dagger$, in general, and it does not depend explicitly on time. It is easy to deduce the following differential equation for $\hat\Psi(t)$:
\be
i\frac{d}{dt}\hat\Psi(t)=H_{nl}(t)\hat\Psi(t), \qquad \mbox{ where }\qquad H_{nl}(t)=H+\frac{1}{2}\langle\hat\Psi(t),(H^\dagger-H)\hat\Psi(t)\rangle.
\label{31}\en
Here "nl" in $H_{nl}$ stands for {\em non-linear}, for obvious reasons. We notice that the original Hamiltonian $H$ for $\Psi(t)$ has to be replaced with another operator which is not only non-linear, but it is also explicitly time-dependent, in general. Of course, if $H=H^\dagger$, then $H_{nl}(t)$ collapses into $H$. It is interesting (but apparently not so useful) to observe that
\be
H_{nl}(t)+H_{nl}^\dagger(t)=H+H^\dagger,
\label{32}\en
which is self-adjoint and time-independent, independently of the details of the system, i.e. on the explicit form of $H$ and of the initial condition $\Psi(0)$.

\vspace{2mm}

Let us consider a certain {\em observable} $X$ of the physical system $\Sc$. For us this is a self-adjoint bounded operator on $\Hil$, $X=X^\dagger$ and $X\in\B(\Hil)$. Notice that $X=X^\dagger$ is more an {\em emotional request} than something essential, considered the fact that we are not assuming this equality even for the Hamiltonian, which is usually considered as the {\em most important} observable of $\Sc$. Moreover, the fact that $X\in\B(\Hil)$ is automatic here, since $\dim(\Hil)<\infty$, while should be really required if $\dim(\Hil)=\infty$, situation which will be considered in a future paper and which requires a very delicate mathematical approach\footnote{This extension is surely very interesting, but we have decided to postpone to a future paper since the situation is already sufficiently complicated as it is, for $\dim(\Hil)<\infty$.}. The main object of the analysis in this and in the next section is the mean value of $X$ on the state $\hat\Psi(t)$:
\be
\xhp=\langle\hat\Psi(t),X\hat\Psi(t)\rangle.
\label{33}\en
Using (\ref{31}) we find that
\be
\frac{d\xhp}{dt}=i\langle\hat\Psi(t),\left(H^\dagger_{nl}(t)X-XH_{nl}(t)\right)\hat\Psi(t)\rangle=\langle\hat\Psi(t),\dhp(X;t)\hat\Psi(t)\rangle,
\label{34}\en
where we have introduced, in analogy with $\delta_\gamma$ in Section \ref{sect2b}, 
\be
\dhp(X;t)=i\left(H^\dagger_{nl}(t)X-XH_{nl}(t)\right)=\delta_\gamma(X)-iX\langle\hat\Psi(t),(H^\dagger-H)\hat\Psi(t)\rangle.
\label{35}\en
We observe that $\dhp(X;t)$ depends, in general, explicitly on time, while $\delta_\gamma(X)$ does not. This is due, of course, to the presence of $\hat\Psi(t)$ in the right-hand side of formula (\ref{35}). This contribution disappears if $H=H^\dagger$. In this case we would get $\dhp(X;t)=\delta_\gamma(X)$, $\forall t\in\mathbb{R}$, and $\delta_\gamma$ is a *-derivation: Proposition \ref{prop3} applies.

When we compare what with have deduced here with the results in Section \ref{sect2b}, we easily understand that $\dhp(.;t)$ is now the major object of our analysis. We start here by defining the following subsets of $B(\Hil)$:
$$
\C_\gamma(H)=\{A\in\B(\Hil):\, \delta_\gamma(X)=0\}, \qquad \C_{\hat\Psi}(H)=\{A\in\B(\Hil):\, \dhp(X;t)=0\}
$$ 
and
$$
\C_{\hat\Psi}^w(H)=\{A\in\B(\Hil):\, \langle\hat\Psi(t),\dhp(X;t)\hat\Psi(t)\rangle=0\},
$$
where $\hat\Psi(t)$ is the same as above: $\hat\Psi(t)=\frac{\Psi(t)}{\norm{\Psi(t)}}$, and $\Psi(t)$ is a solution of $i\dot\Psi(t)=H\Psi(t)$. It is clear that the set $\C_\gamma(H)$, which is the set of all the operators we have called $\gamma$-symmetries in Section \ref{sect2b}, is not a subset of neither $\C_{\hat\Psi}(H)$ nor $\C_{\hat\Psi}^w(H)$, in general. In fact, let us fix an element $X\in\C_\gamma(X)$. Then, using (\ref{35}), $\dhp(X;t)=-iX\langle\hat\Psi(t),(H^\dagger-H)\hat\Psi(t)\rangle$. Taking the mean value of this equality on $\hat\Psi(t)$, after some manipulation, we deduce that
\be
\frac{d}{dt} \left(\|\Psi(t)\|^2\xhp\right)=0, \qquad \Rightarrow \qquad \xhp=\frac{\|\Psi(0)\|^2}{\|\Psi(t)\|^2}\,x_{\hat\Psi}(0)= \frac{x_{\hat\Psi}(0)}{\|\Psi(t)\|^2},
\label{37}\en
if we assume that $\Psi(t)$ is normalized at $t=0$. It is clear then that, except for very particular situations, $\|\Psi(t)\|^2$ depends explicitly on time, and so does $\xhp$: $\xhp$ is not constant, in general! Hence $X\notin \C_{\hat\Psi}^w(H)$ and, obviously, $X\notin \C_{\hat\Psi}(H)$ either.

On the other hand, it is easy to understand that $\C_{\hat\Psi}(H)\subset\C_{\hat\Psi}^w(H)$. Indeed, if $X\in \C_{\hat\Psi}(H)$ then $\dhp(X;t)=0$. Therefore, its mean value on any vector, and in particular on $\hat\Psi(t)$, is zero. Hence $X\in\C_{\hat\Psi}^w(H)$. The fact that the inclusion is proper follows from noticing that $\1\notin\C_{\hat\Psi}(H)$, but $\1\in\C_{\hat\Psi}^w(H)$. Indeed we have $$\dhp(\1;t)=\delta_\gamma(\1)-i\langle\hat\Psi(t),(H^\dagger-H)\hat\Psi(t)\rangle=i(H^\dagger-H)-i\langle\hat\Psi(t),(H^\dagger-H)\hat\Psi(t)\rangle,$$
which is clearly (in general) non zero. However, when we compute its mean value on $\hat\Psi(t)$, we get $\langle\hat\Psi(t),\dhp(\1;t)\hat\Psi(t)\rangle=0$. Hence our claim follows. 

This is in agreement with the fact that if we take $X=\1$ in (\ref{34}), and we put $\id_{\hat\Psi(t)}=\langle\hat\Psi(t),\1\hat\Psi(t)\rangle$, both sides of (\ref{34}) are zero: the RHS since $\langle\hat\Psi(t),\dhp(\1;t)\hat\Psi(t)\rangle=0$, and the LHS because $\frac{d\id_{\hat\Psi(t)}}{dt}=\frac{d\|\hat\Psi(t)\|^2}{dt}=0$, since $\|\hat\Psi(t)\|=1$ for all $t$.

 The interesting conclusion is that, while $\delta_\gamma$ does not annihilate the identity operator, even weakly, $\dhp(\1;t)$ has always mean value zero on the vector $\hat\Psi(t)$, and therefore on its un-normalized version $\Psi(t)$. We introduce now the following definition, related in same natural way to the analogous definition for $\gamma$-symmetries.
\begin{defn}\label{def1}
	$A\in\B(\Hil)$ is called a weak $\hat\Psi$-integral of motion if $A\in\C_{\hat\Psi}^w(H)$. $B\in\B(\Hil)$ is called a  $\hat\Psi$-integral of motion if $A\in\C_{\hat\Psi}(H)$. 
	\end{defn}
From what we have seen before, the identity operator is not, in general, a $\hat\Psi$-integral of motion, but it is a weak $\hat\Psi$-integral of motion. It is also clear that all the $\hat\Psi$-integrals of motion are also weak $\hat\Psi$-integrals of motion. The main difference between these  $\hat\Psi$-integrals of motion (weak or not) and the $\gamma$-symmetries is that these latter do not depend on the state of the system, while the $\hat\Psi$-integrals of motion do.

It is possible to find a necessary condition which is satisfied by any $X\in\B(\Hil)$ which is a weak $\hat\Psi$-integrals of motion, i.e. such that $\xhp=\langle\hat\Psi(t),X\hat\Psi(t)\rangle=x_{\hat\Psi}(0)$, see (\ref{33}). Indeed, using (\ref{34}) and multiplying the equation for $\|\Psi(t)\|^2$, we get that if  $\frac{d\xhp}{dt}=0$, then
\be
\langle\Psi(t),\delta_\gamma(X)\Psi(t)\rangle=i\,x_{\hat\Psi}(0)\langle\Psi(t),(H^\dagger-H)\Psi(t)\rangle.
\label{37b}\en
Notice that the mean values in both sides are taken on the (un-normalized) vector $\Psi(t)$, and that $x_{\hat\Psi}(0)$ appears in the RHS since, by assumption, $\xhp$ is constant in time.

The map $\dhp$ shares with $\delta_\gamma$ the property of being stable under the adjoint map. More in details, if $X\in \C_{\hat\Psi}(H)$, then $X^\dagger\in\C_{\hat\Psi}(H)$, too. Indeed we can check that, for all $A\in\B(\Hil)$,
\be
\dhp(A;t)^\dagger=\dhp(A^\dagger;t),
\label{38}\en
$\forall t\in\mathbb{R}$.
Now, if $X\in \C_{\hat\Psi}(H)$, then $\dhp(X;t)=0$. But, using (\ref{38}), we have $\dhp(X^\dagger;t)=\dhp(X;t)^\dagger=0$. Hence $X^\dagger\in\C_{\hat\Psi}(H)$ as well, as we had to check.

Of course, the same stability is also satisfied by $\C_{\hat\Psi}^w(H)$: if $X\in \C_{\hat\Psi}^w(H)$, then we also have $X^\dagger\in\C_{\hat\Psi}^w(H)$.

Another analogy between the sets  $\C_{\gamma}(H)$ and $\C_{\hat\Psi}(H)$ is the following: we have seen that if $X\in\B(\Hil)$ is a $\gamma$-symmetry, then $X_H=XH$ is a $\gamma$-symmetry as well. Stated differently: if $X$ satisfies $\delta_\gamma(X)=0$, then we find that $\delta_\gamma(X_H)=0$. It is also possible to check, using the stability of $\delta_\gamma$ under the adjoint, that $\delta_\gamma(X_H^\dagger)=\delta_\gamma(H^\dagger X^\dagger)=0$. 

Similar conclusions can be deduced for $\dhp$: first we observe that, for all $A\in\B(\Hil)$, we have the following identities:
\be
\dhp(H^\dagger A)=H^\dagger\dhp(A), \qquad \mbox{ and }\qquad \dhp(A H)=\dhp(A)H,
\label{39}\en
which extend similar identities for $\delta_\gamma$. Hence the following result holds:
\begin{lemma}
	If $X\in \C_{\hat\Psi}(H)$ then we also have $H^\dagger X,XH\in\C_{\hat\Psi}(H)$, together with $X^\dagger$, $X^\dagger H$ and $H^\dagger X^\dagger$.
\end{lemma}

\vspace{2mm}

{\bf Remark:--}
Except that in some special case, we prefer not to use here the terms {\em symmetries} or {\em integrals of motion} since what we have seen before for the identity operator suggests that the elements of $\C_{\gamma}(H)$ are not necessarily constants of motions, while some suitable mean value of them could turn out to be time-independent.

\vspace{2mm}

Another useful property of $\dhp$ is its linearity: $\dhp(\alpha A+\beta B;t)=\alpha\dhp(A;t)+\beta\dhp(B;t)$, $\forall\alpha,\beta\in\mathbb{C}$, $A,B\in\B(\Hil)$. Also, $\dhp$ is norm continuous. This is a consequence of the following inequality:
\be
\|\dhp(A;t)\|\leq4\|H\|\|A\|,
\label{310}\en
which follows from (\ref{35}), from the Schwartz inequality, and from the fact that $\|\delta_\gamma(A)\|\leq2\|H\|\|A\|$, as we have seen in Section \ref{sect2b}. Hence, as  $\delta_\gamma$, also $\dhp$ is norm-continuous, since $\|\dhp(X_n;t)-\dhp(X;t)\|\leq4\|H\|\|X_n-X\|\rightarrow0$, for any sequence $\{X_n\}$ norm-convergent to $X$. The same inequality in (\ref{310}) implies that the series $\sum_{k=0}^{\infty}\frac{t^k\dhp^k(X;t)}{k!}$ is norm convergent for all $X\in\B(\Hil)$. The notation is essentially the same as that in Section \ref{sect2b}. However, with respect to what we have seen before for $\delta_\gamma$, we observe two major differences: $\dhp$ is, in general, time-dependent (while $\delta_\gamma$ is not!) and we are not saying here to which operator the series converge. Of course, everything simplifies if $H=H^\dagger$, but it is much more complicated in the situation which is more interesting for us, i.e. when $H\neq H^\dagger$.

\subsection{A special case}\label{sect3a}

The fact that $\sum_{k=0}^{\infty}\frac{t^k\dhp^k(X;t)}{k!}$ converges, together with the results in Section \ref{sect2b}, makes us confident that this series might describe, in some way, the time evolution of $X$ when smeared on a state $\hat\Psi(t)$. We discuss here the simple situation in which $\hat\Psi(0)$ is an eigenstate of $H$. In other words we assume that
\be
\hat\Psi(0)=\varphi_{k_0}, \qquad \mbox{ where }\qquad H \varphi_{k_0}=E_{k_0}\varphi_{k_0}.
\label{311}\en
It is useful to stress that we are not assuming that $E_{k_0}$ is real. Indeed we write  $E_{k_0}=E_{k_0}^{(r)}+iE_{k_0}^{(i)}$. Because of the nature of $\hat\Psi(t)$, which is normalized for all $t\geq0$, we need to assume here that $\|\varphi_{k_0}\|=1$. Hence $\hat\Psi(t)=e^{-iE_{k_0}^{(r)}t}\varphi_{k_0}$, and we conclude that
$$
\langle\hat\Psi(t),(H^\dagger-H)\hat\Psi(t)\rangle=-2iE_{k_0}^{(i)},
$$
which is independent of $t$. Hence (\ref{35}) becomes much simpler:
\be
\dhp(X;t)=\delta_\gamma(X)-2E_{k_0}^{(i)}X,
\label{312}\en
which is independent of $t$, too. If we now introduce the following quantities:
\be
H_{k_0}=H-E_{k_0}\1; \qquad \beta_{\hat\Psi}^t(X):=\sum_{k=0}^{\infty}\frac{t^k\dhp^k(X;t)}{k!}, \qquad \gamma_{\hat\Psi}^t(X):=e^{iH_{k_0}^\dagger t}Xe^{-iH_{k_0}t},
\label{313}\en
we can prove the following result, which allows us to answer, at least in this very special case, to our previous question: what is the sum of $\sum_{k=0}^{\infty}\frac{t^k\dhp^k(X;t)}{k!}$? Indeed we have
\begin{prop}\label{prop}
With the above definitions we find that
\be\beta_{\hat\Psi}^t(X)=\gamma_{\hat\Psi}^t(X),
\label{314}\en
for all $t\geq0$ and for all $X\in\B(\Hil)$.
\end{prop}

\begin{proof}
	We start noticing that both $\beta_{\hat\Psi}^t(X)$ and $\gamma_{\hat\Psi}^t(X)$ are well defined. This is because the series for $\beta_{\hat\Psi}^t(X)$ is norm converging, while $\gamma_{\hat\Psi}^t(X)$ is simply the product of three bounded operators. Using the fact that $\dhp^k(X;t)$ does not depend explicitly on $t$, it is now possible to check that
	$$
\frac{d}{dt}\left[\beta_{\hat\Psi}^t(X)\right]	\frac{d}{dt}\left[\sum_{k=0}^{\infty}\frac{t^k\dhp^k(X;t)}{k!}\right]=\dhp\left[\beta_{\hat\Psi}^t(X)\right].
	$$ 
	It is also possible to deduce that
		$$
	\frac{d}{dt}\left[\gamma_{\hat\Psi}^t(X)\right]=\dhp\left[\gamma_{\hat\Psi}^t(X)\right].
	$$ 
	Therefore, putting $\Gamma_{\hat\Psi}^t(X)=\beta_{\hat\Psi}^t(X)-\gamma_{\hat\Psi}^t(X)$, we deduce that $\Gamma_{\hat\Psi}^0(X)=X-X=0$ and that 
	$$
	\frac{d}{dt}\left[\Gamma_{\hat\Psi}^t(X)\right]=\dhp\left[\Gamma_{\hat\Psi}^t(X)\right].
	$$ 
	Hence $\Gamma_{\hat\Psi}^t(X)\equiv0$ is a solution of this equation, with the correct initial condition. Moreover, see \cite{bagmor}, this is the only solution, as we had to prove.
	
\end{proof}

The fact that this is a particularly simple situation can also be seen just checking that, using (\ref{312}), $X\in\C_{\hat\Psi}^w(H)$: $ \langle\Psi(t),\dhp(X;t)\Psi(t)\rangle= \langle\hat\Psi(t),\dhp(X;t)\hat\Psi(t)\rangle=0$, which implies that $\frac{d\xhp}{dt}=0$. Stated differently, assuming that $\Psi(0)$ is an eigenstate of $H$, each $X\in\B(\Hil)$ is a  weak $\hat\Psi$-integral of motion. Of course, however, this does not imply necessarily that $X$ does not evolve in time as operator.

\vspace{1mm}

{\bf Remarks:--} (1) It is now quite natural to call $\gamma_{\hat\Psi}^t(X)$ the {\em time evolution of $X$}. But this makes sense in the settings considered here. The possibility of extending this result to more general choices of $\hat\Psi(0)$ remains open.

\vspace{1mm}

(2) The proof of Proposition \ref{prop} could be extended, we believe, to the more general case in which $\dhp(X;t)$ does not depend explicitly on $t$. This is work in progress.

\vspace{2mm}

Notice that  $\gamma_{\hat\Psi}^t\1)$ satisfies the following equalities:
\be
\langle\hat\Psi(0),\gamma_{\hat\Psi}^t(\1)\hat\Psi(0)\rangle=1, \qquad \langle\hat\Psi(0),\dhp(\1;t)\hat\Psi(0)\rangle=0,
\label{315}\en
even if $\gamma_{\hat\Psi}^t(\1)\neq\1$ and $\dhp(\1;t)\neq0$: this $\hat\Psi$-dynamics behave as expected on $\1$ only weakly, but not at an operatorial level. This is in line with the role of the time evolution of the identity operator we have considered in Section \ref{sect2b}, see Proposition \ref{prop3}. However, it should also be emphasized that, given $X,Y\in\B(\Hil)$, in general we have
$$
\langle\hat\Psi(0),\gamma_{\hat\Psi}^t(XY)\hat\Psi(0)\rangle\neq\langle\hat\Psi(0),\gamma_{\hat\Psi}^t(X)\gamma_{\hat\Psi}^t(Y)\hat\Psi(0)\rangle,
$$
which shows that $\gamma_{\hat\Psi}^t$ is not  automorphism, even when considered at a weakly level, and not even for this simple situation. This no-go result opens the way to a very natural, and difficult, problem: even if we are sure that $\sum_{k=0}^{\infty}\frac{t^k\dhp^k(X;t)}{k!}$ is norm convergent for all $X\in\B(\Hil)$, and we know that, in some cases, this series converge to the time evolution of $X$ (for instance if $H=H^\dagger$, or when we work with $\Psi(t)$ rather than with $\hat\Psi(t)$), are we sure that the series define the time evolution of $X$ also in other situations? This is, so far, an open (and, we believe, quite interesting) problem.

\section{An example}\label{sect4}

In Section \ref{sect3} we have considered conditions which ensure that, see (\ref{34}), the mean value of some observable $X$ on the normalized solution of the Schr\"odinger equation is constant in time. In particular, in Section \ref{sect3a} we have considered a very simple case, where $\hat\Psi(0)$ is an eigenstate of $H$. In this section we want to describe an example, first considered in \cite{bag2020}, in which a conserved quantity is found. We refer to \cite{bag2020} for all the details on this model, and for the rationale to introduce it, and its role in Decision Making. Here, after a very concise review of the results, we will only show that this conserved quantity fits Definition \ref{def1}.

 The  model we want to consider in this section is a 3 { fermionic} agents system $\Sc$. This means that each agent (i.e. each {\em degree of freedom}) has only two allowed levels, and that they can go from one level to the other by means of some operator $b_j$ and $b_j^\dagger$, $j=1,2,3$,
 which satisfy the following CAR:
 \be
 \{b_k,b_j^\dagger\}=\delta_{k,j}\1, \qquad \qquad b_j^2=0,
 \label{41}\en
 $j,k=1,2,3$. Here $\1$ is the identity operator on the Hilbert space of the system, $\Hil=\mathbb{C}^8$. Next we use these operators, and their adjoints, to construct an orthonormal (o.n.) basis for $\Hil$. We start with $\varphi_{000}=(1\,\, 0\,\, 0\,\, 0\,\, 0\,\, 0\,\, 0\,\, 0 )^T$, (here $T$ is the transpose of the vector). This vector is such that $b_j\varphi_{000}=0$, $j=1,2,3$. Then we introduce
 $$
 \varphi_{100}=b_1^\dagger\varphi_{000}, \quad \varphi_{010}=b_2^\dagger\varphi_{000}, \quad \varphi_{001}=b_3^\dagger\varphi_{000}, \quad \varphi_{110}=b_1^\dagger b_2^\dagger\varphi_{000}, 
 $$
 $$
 \varphi_{101}=b_1^\dagger b_3^\dagger\varphi_{000}, \quad\varphi_{011}=b_2^\dagger b_3^\dagger\varphi_{000},\quad \varphi_{111}=b_1^\dagger b_2^\dagger b_3^\dagger\varphi_{000},
 $$
 The set $\F_\varphi=\{\varphi_{ijk},\,i,j,k=0,1\}$ is an o.n. basis of $\Hil$.
We assume that the dynamics of the system $\Sc$ is driven by the Hamiltonian, $$H=b_1^\dagger(\lambda b_2+\mu b_3),$$ with $\lambda$ and $\mu$ positive quantities, \cite{bag2020}. Because of the nature of the ladder operators $b_j$ and $b_j^\dagger$, the action of $H$ increases $n_1(t)$, while decreasing both $n_2(t)$ and $n_3(t)$, where $N_j=b_j^\dagger b_j$ and $n_j(t)$ is found as in (\ref{33}):
\be
n_j(t)=\langle\hat\Psi(t),N_j\hat\Psi(t)\rangle,
\label{42}\en
$j=1,2,3$. The computation of these mean values is easy, \cite{bag2020}, mainly because the operator $U(t)=e^{-iHt}$ can be rewritten as $U(t)=\1-iHt$, due to the fact that $H^2=0$, because of (\ref{41}).

 Let us now assume, to begin with, that $\hat\Psi(0)=\varphi_{011}$. This vector corresponds to $n_1(0)=0$ and $n_2(0)=n_3(0)=1$. If we now use formula (\ref{42}) we find
$$
 n_1(t)=\frac{(\mu^2+\lambda^2)t^2}{1+(\mu^2+\lambda^2)t^2}, \quad n_2(t)=\frac{1+\mu^2t^2}{1+(\mu^2+\lambda^2)t^2}, \quad n_3(t)=\frac{1+\lambda^2t^2}{1+(\mu^2+\lambda^2)t^2},   
$$
 which exhibit the expected behavior: $n_1(t)$ increases from 0 to 1, $n_2(t)$ decreases from 1 to $\frac{\mu^2}{\mu^2+\lambda^2}$ and $n_3(t)$ decreases from 1 to $\frac{\lambda^2}{\mu^2+\lambda^2}$. What is particularly interesting for us, here, is that the sum of $n_1(t)+n_2(t)+n_3(t)$ is always equal to 2: \be n_1(t)+n_2(t)+n_3(t)=2=n_1(0)+n_2(0)+n_3(0). \label{43}\en
 A similar result is deduced if  we consider a different initial condition for $\Sc$. In fact, if we take $\hat\Psi(0)=\varphi_{010}$, repeating the same computations, we get
$$
 n_1(t)=\frac{\lambda^2t^2}{1+\lambda^2t^2}, \quad n_2(t)=\frac{1}{1+\lambda^2t^2}, \quad n_3(t)=0,
$$
 which are again in agreement with the fact that $n_1(0)=n_3(0)=0$ and $n_2(0)=1$, and with the property of $H$ to destroy a state with $n_3(0)=0$, because of $b_3$, to increase $n_1(t)$ and to decrease $n_2(t)$. In this case we find that
 \be n_1(t)+n_2(t)+n_3(t)=1=n_1(0)+n_2(0)+n_3(0), \label{44}\en
 showing that the mean value on $\hat\Psi(t)$ of the total number operator  is again a constant of motion. We stress that is both these cases $\hat\Psi(0)$ is not an eigenstate of $H$. Hence what we are discussing here is different from what we have seen in Section \ref{sect3a} and, we believe, quite interesting.

 Similar considerations can be repeated with other choices of $\Psi(0)$, of course.  
 
 We will now show that these two cases agree with what deduced before. In particular, we will show that the total number operator of $\Sc$ is an element of $\C_{\hat\Psi}^w(H)$. 
 
 To do this, we first need to compute $\delta_\gamma(N)$, where $N=N_1+N_2+N_3$. This is because  $\delta_\gamma(N)$ is the first contribution to $\dhp(N;t)$, see (\ref{35}). Notice that this is a state-independent term. Using (\ref{41}) we deduce that
 \be
 \delta_\gamma(N)=i\lambda\left(b_2^\dagger b_1-b_1^\dagger b_2\right)\left(\1+N_3\right)+i\mu\left(b_3^\dagger b_1-b_1^\dagger b_3\right)\left(\1+N_2\right).
  \label{45}\en
  Next we observe that, because of the already cited identity $H^2=0$, 
  $$
  \hat\Psi(t)=\frac{(\1-iHt)\Psi(0)}{\norm{(\1-iHt)\Psi(0)}}=\frac{\varphi_{011}-it(\lambda\varphi_{101}+\mu\varphi_{110})}{1+t^2(\lambda^2+\mu^2)},
  $$
 if $\Psi(0)=\varphi_{011}$. With this choice the second term in (\ref{35}) returns, using again (\ref{41}),
 $$
 \langle\hat\Psi(t),(H^\dagger-H)\hat\Psi(t)\rangle=\frac{-2it(\lambda^2-\mu^2)}{1+t^2(\lambda^2+\mu^2)}.
 $$
 It is now clear that, see again (\ref{35}), $$\dhp(N;t)=i\lambda\left(b_2^\dagger b_1-b_1^\dagger b_2\right)\left(\1+N_3\right)+i\mu\left(b_3^\dagger b_1-b_1^\dagger b_3\right)\left(\1+N_2\right)-iN\frac{-2it(\lambda^2-\mu^2)}{1+t^2(\lambda^2+\mu^2)},$$
 which is different from zero, in general. However, if we compute its mean value on $\hat\Psi(t)$, we get
 $$
 \langle\hat\Psi(t),\dhp(N;t)\hat\Psi(t)\rangle=0,
 $$
 in agreement with the results in Section \ref{sect3} and with what deduced above in (\ref{43})
 
 \vspace{2mm}
 
 If we consider another initial condition, and we start with $\hat\Psi(0)=\varphi_{010}$, then we deduce that $$
 \hat\Psi(t)=\frac{\varphi_{010}-it\lambda\varphi_{100}}{1+t^2\lambda^2}, \qquad \langle\hat\Psi(t),(H^\dagger-H)\hat\Psi(t)\rangle=\frac{-2it\lambda^2}{1+t^2\lambda^2},
 $$
 and again  $ \langle\hat\Psi(t),\dhp(N;t)\hat\Psi(t)\rangle=0,
 $
 in agreement with the identity in (\ref{44}).
 
 This example, originally considered in a completely different field of research, is really motivating for us, since it suggests, in a very concrete situation, that the role of constant quantities in presence of non-Hermitian Hamiltonians is still to be fully understood, and not trivial at all.

\section{Conclusions}\label{sectconcl}

In this paper we have discussed the Heisenberg dynamics for a non-Hermitian Hamiltonian as deduced using normalized wave functions obeying a non-linear version of the Schr\"odinger equation. As often, nonlinearities makes an already complicated situation even more complicated. But, exactly for this reason, it also becomes quite interesting both for its mathematical aspects and because of its possible applications. In this perspective, we have seen that it is easy, and unexpected, to find quantities which remain constant during the time evolution of some system, even if when the system is driven by manifestly non-Hermitian Hamiltonians. The possible relation of these constant quantities with symmetries of the system is still to be understood. In fact, this paper open the way to several questions that should be investigated, a part the role (and a {\em smart} definition) of symmetries. In fact, we would like to understand if the sum of the series $\sum_{k=0}^{\infty}\frac{t^k\dhp^k(X;t)}{k!}$, which as we have seen being surely norm convergent, can be interpreted as the time evolution of $X$, and which is the operatorial differential equation (which should replace or extend the Heisenberg equation of motion) obeyed by this quantity. The sets $\C_{\hat\Psi}(H)$ and $\C_{\hat\Psi}^w(H)$ also deserve a deeper analysis, in view of their relations with $\dhp$.

Last but not least, extensions to infinite dimensional Hilbert spaces are relevant, and absolutely non trivial, in view of the presence (in general) of unbounded operators with their related domain issues. These are part of our future projects.

\section*{Acknowledgements}

 The author expresses his gratitude to Yogesh Joglekar for many interesting discussions  on $\gamma$-symmetries. The author acknowledges partial financial support from Palermo University and from G.N.F.M. of the INdAM. 

\section*{Fundings}

This work has also been partially supported by the PRIN grant {\em Transport phenomena in low dimensional
	structures: models, simulations and theoretical aspects}- project code 2022TMW2PY - CUP B53D23009500006, and partially by  project ICON-Q, Partenariato Esteso NQSTI - PE00000023, Spoke 2.

\section*{Conflicts of interest}

There are no conflicts of interest.

\section*{Availability of data and material}

Not applicable.

\section*{Code availability}

Not applicable.

\section*{Authors' contributions}

Not applicable.


\begin{thebibliography}{99}
	
%
	
	\bibitem{benbook} C.M. Bender, {\em PT Symmetry in Quantum and Classical Physics}, World Scientific, Singapore,  (2019)
	
	
	\bibitem{specissue2012} C. M. Bender, A. Fring, U. G\"nther, H. Jones Eds, {\em Special issue on quantum physics with non-Hermitian operators}, J. Phys. A: Math. and Ther., {\bf 45} (2012)



\bibitem{bagprocpa} F. Bagarello, R. Passante, C. Trapani, {\em Non-Hermitian Hamiltonians in Quantum Physics;
	Selected Contributions from the 15th International Conference on Non-Hermitian
	Hamiltonians in Quantum Physics, Palermo, Italy, 18-23 May 2015}, Springer (2016)



\bibitem{specissue2021} C. M. Bender, A. Fring, F. Correa Eds, {\em Proceedings for "Pseudo-Hermitian Hamiltonians in Quantum Physics"}, Journal of Physics: Conference series, C. M. Bender, F. Correa and A. Fring Eds., {\bf 2038}, 012001, (2021)


\bibitem{bagbookPT} F. Bagarello, J. P. Gazeau, F. H. Szafraniec e M. Znojil Eds., {\em Non-selfadjoint operators in quantum physics: Mathematical aspects}, John Wiley and Sons (2015)

	\bibitem{bagspringer} F. Bagarello, {\em Pseudo-bosons and their coherent states}, Springer (2022)
	
			\bibitem{bagAoP1} F. Bagarello, {\em Some results on the dynamics and  transition probabilities for non self-adjoint hamiltonians},  Ann. of Phys., {\bf 356}, 171-184 (2015)	
	
	
	\bibitem{bagAoP2}  F. Bagarello, {\em Transition probabilities for non self-adjoint Hamiltonians in infinite dimensional Hilbert spaces},
	Ann. of Phys., {\bf 362}, 424-435 (2015)
	
	
	
	\bibitem{bag2022} F. Bagarello, {\em Heisenberg dynamics for non self-adjoint Hamiltonians: symmetries and derivations},    MPAG, {\bf 26}, 1-15 (2023)
	
	
	\bibitem{yog1} F. Ruzicka, K. S. Agarwal, Y. N. Joglekar, {\em Conserved quantities, exceptional points, and antilinear symmetries in non-Hermitian systems},  J. Phys.: Conf. Ser. {\bf 2038} 012021 (2021)
	
	
	
	\bibitem{sato} K. Kawabata, K. Shiozaki,  M. Ueda, M. Sato, {\em Symmetry and Topology in Non-Hermitian Physics}, Phys. Rev. X {\bf 9}, 041015 (2019)
	
	
		\bibitem{muga} M. A. Sim\'on Mart\'inez, A. Buend\'ia, and J. G. Muga, {\em Symmetries and invariants for non-Hermitian Hamiltonians}, Mathematics 2018, 6(7), 111; https://doi.org/10.3390/math6070111
	
	\bibitem{sim} K. Sim, N. Defenu, P. Molignini, R. Chitra, {\em Generalized symmetry in non-Hermitian systems}, 	Phys. Rev. Research {\bf 7}, 013325  (2025)

	\bibitem{lisok}  A. L. Lisok, A. V. Shapovalov, A. Yu. Trifonov, {\em Symmetry and Intertwining Operators for the Nonlocal Gross-Pitaevskii Equation}, SIGMA {\bf 9}, 066, 21 pages (2013)

	\bibitem{bag2020} F. Bagarello,  {\em One-directional quantum mechanical dynamics and an application to decision making},  Physica A, {\bf 537}, 122739, (2020)



	\bibitem{sgh} F. G. Scholtz, H. B. Geyer, F. J. Hahne, {\em
		Quasi-Hermitian operators in quantum mechanics and the variational principle}, Ann. Phys. {\bf 213}, 74-101 (1992)
	

	
	
	
	\bibitem{br} O. Bratteli and D.W. Robinson, {\em Operator
		algebras and Quantum statistical mechanics 1}, Springer-Verlag, New
	York, (1987)
	
	
		\bibitem{chri}O. Christensen, {\em An Introduction to Frames and Riesz Bases}, Birkh\"auser, Boston, (2003)
	
	
	
	\bibitem{heil} C. Heil, {\em A basis theory primer: expanded edition}, Springer, New York, (2010)
	
	
	
	
	
	
	
	\bibitem{bagbook} F. Bagarello, {\em Deformed canonical (anti-)commutation relations and non hermitian hamiltonians}, in {Non-selfadjoint operators in quantum physics: Mathematical aspects}, F. Bagarello, J. P. Gazeau, F. H. Szafraniek and M. Znojil Eds., John Wiley and Sons Eds., (2015)

	\bibitem{gant} F. R. Gantmacher, {\em The theory of matrices}, Chelsea Publishing Company, New York, (1959)
	
	
	
	
	
	
	
	\bibitem{bagmor} F.Bagarello, G. Morchio, {\em Dynamics of mean field spin models from
		basic results in abstract  differential equations}, J. Stat. Phys.
	{\bf 66}, 849-866 (1992)
	
	
	
	
	
	
	
	
%
%
%
	
	
	
	
%
%
%
%
%
%
%
%
%
%
%
%
%
%
%
%
%
%
%
%
%
%
%
%
%
%
%
%
%
%
%
%
%
%
%
%
%
	
\end{thebibliography}
\end{document}